\documentclass[cis]{ipart}

\RequirePackage[OT1]{fontenc}
\RequirePackage{amsthm,amsmath,amssymb}
\RequirePackage[numbers,square]{natbib} %
\RequirePackage[colorlinks]{hyperref}
\usepackage{enumerate}
\newtheorem{thm}{Theorem}
\newtheorem{lem}{Lemma}

\newcommand{\EE}{\mathbb{E}}
\renewcommand{\d}{\operatorname{d}\!}

\pubyear{2016}
\volume{PP}
\issue{PP}
\firstpage{1}
\lastpage{5}

\begin{document}

\begin{frontmatter}
\title{Monotonicity of Entropy and Fisher Information:\\ A Quick Proof via Maximal Correlation}
\runtitle{Monotonicity of Entropy and Fisher Information: A Quick Proof  }
\runtitle{Monotonicity of Entropy and Fisher Information: A Quick Proof }

\begin{aug}
\author{\fnms{Thomas} \snm{A.~Courtade}\ead[label=e1]{courtade@berkeley.edu}}
\address{Department of Electrical Engineering and Computer Sciences\\
University of California, Berkeley\\
Berkeley, CA 94720\\
USA.\\
\printead{e1}}

\thankstext{t1}{This work supported by the Center for Science of Information (CSoI), a NSF Science and Technology Center, under grant agreement CCF-0939370.}
\runauthor{Thomas~A.~Courtade}

\affiliation{Some University and Another University}

\end{aug}

\begin{abstract}
A simple proof is given for the monotonicity of  entropy and Fisher information associated to  sums of i.i.d.\ random variables.  The proof relies on a characterization of maximal correlation for partial sums due to Dembo, Kagan and Shepp. 
\end{abstract}

\begin{keyword}
\kwd{Entropy}
\kwd{Fisher information}
\kwd{Maximal Correlation}
\end{keyword}
\end{frontmatter}

\section{Introduction}
Assume throughout that $X$ is a random variable with density $f$ and finite variance.  The entropy $h(X)$ and, under mild regularity conditions on $f$, the Fisher information $J(X)$ are defined  via 
\begin{align*}
&h(X) = -\EE[\log f(X)],  &J(X) = \EE\left[ \rho^2_X(X)\right],
\end{align*}
where $\rho_X = f'/f$ denotes the score function associated to $X$.

Let $X_1, X_2, \dots$ be i.i.d.\ copies of $X$ and define 
$S_n = X_1 + \cdots + X_n$, $n\geq 1,$ and its standardized counterpart $U_n = \tfrac{1}{\sqrt{n}}S_n$.  Two celebrated results established by {Artstein}, Ball, Barthe and Naor  \cite{artstein2004solution} are:  
\begin{enumerate}[i)]
\item the entropies $h(U_n)$ are non-decreasing in $n$; and 
\item the Fisher informations $J(U_n)$ are non-increasing in $n$.  
\end{enumerate}
In other words, the respective central limit theorems for entropy \cite{barron1986entropy} and Fisher information \cite{johnson2004fisher} enjoy monotone convergence (the latter holding under mild regularity conditions on $f$). 

The aim of this  note is to point out a simple and brief proof of these facts using a characterization of maximal correlation for sums of i.i.d.\ random variables. 

\section{Monotonicity of Fisher Information and Entropy}
The maximal correlation associated to a random pair $X,Y$ is defined (in one of its equivalent forms) as 
\begin{align}
r^2(X;Y) = \sup_{\vartheta} \frac{\EE [ \left|\EE[\vartheta(X)|Y] \right|^2  ] }{ \EE[ \left| \vartheta(X)\right|^2  ]  } ,\label{maxCorrDef}
\end{align}
where the supremum is  over all non-constant, real-valued functions $\vartheta$ with $\EE \vartheta(X)=0$.  An unexpected property enjoyed by $r^2$, discovered by Dembo, Kagan and Shepp, is that  $r^2(S_m;S_n) = m/n$ for $1\leq m \leq n$. A  brief proof of the Dembo-Kagan-Shepp identity has been recently obtained by Kamath and Nair  using  information-theoretic arguments \cite{kamath2015strong}.  

As a consequence,   if $\vartheta:\mathbb{R}\to \mathbb{R}$ satisfies $\EE \vartheta(S_m)=0$ and is non-constant, then definition \eqref{maxCorrDef} combined with the Dembo-Kagan-Shepp identity yields 
\begin{align}
& \EE [ \left|\EE[\vartheta(S_m)|S_n] \right|^2  ]   \leq  \frac{m}{n}  \EE[ \left| \vartheta(S_m)\right|^2  ]   &1\leq m\leq n. \label{maxCorr}
\end{align}

The contraction \eqref{maxCorr} is the first ingredient in our proof, and we shall require one more: the  behavior of  score functions under convolution, first noted by Stam \cite{stam1959some}.
\begin{lem}\label{scoreFunction}
Let $U,V$ be  independent random variables with smooth densities and put $W = U+V$.  If $\rho_U$ and $\rho_W$ denote the score functions of $U$ and $W$ respectively, then
\begin{align}
\rho_{W}(w) = \EE[\rho_{U}(U) | W=w ] .\label{scoreFunctionConv}
\end{align}
\end{lem}
Identity \eqref{scoreFunctionConv} is proved by exchanging orders of differentiation and integration, and is justified by smoothness of densities (e.g., \cite[Lemma 1.20]{johnson2004information}).

\begin{thm}[Monotonicity of Fisher Information] \label{thm:J}
Assume $X$ has  smooth density.  For $1\leq m \leq n$, $J(U_{n}) \leq J(U_{m})$.
\end{thm}
\begin{proof}
By Lemma \ref{scoreFunction}, we have $\rho_{S_n}(s) = \EE[\rho_{S_m}(S_m) | S_n=s ]$.  Moreover, $\EE \rho_{S_m}(S_m)=0$, so that $\vartheta = \rho_{S_m}$ is a valid choice in  \eqref{maxCorr}.  Hence, from the definition of Fisher information and \eqref{maxCorr}, we conclude  
\begin{align}
J(S_n) = \EE[\rho_{S_n}^2(S_n) ] &= \EE[ \left| \EE\left[ \rho_{S_{m}}(S_{m})|S_n  \right]\right|^2 ]\notag \\
& \leq \frac{m}{n}\EE[\rho^2_{S_{m}}(S_{m})] = \frac{m}{n} J(S_{m}).\label{mainInequality}
\end{align}
Noting  the scaling property $\alpha^2 J\left( {\alpha} X\right) = J(X)$   finishes the proof. 
\end{proof}
Exactly as in \cite{artstein2004solution}, the entropy counterpart follows directly from a standard semigroup argument, which derives from Stam's seminal paper \cite{stam1959some}.  We include it for completeness. 

\begin{thm}[Monotonicity of Entropy]\label{thm:H}
 For $1\leq m \leq n$, $h(U_{m}) \leq h(U_{n})$.
\end{thm}
\begin{proof}
For a random variable $Z$ with unit variance, define the Ornstein-Uhlenbeck evolutes $Z_t = e^{-t} Z + (1-e^{-2t})^{1/2}G$, where $G$ is standard normal independent of $Z$. Note that $Z_t$ has smooth density for $t>0$.  By de Bruijn's identity (e.g.,  \cite{stam1959some},\cite[Appendix C]{johnson2004information}),
\begin{align}
h(G) - h(Z)= \int_{0}^{\infty}(J(Z_t) -1) \d t.
\end{align} 
Using these facts, we find that Theorem \ref{thm:H} follows from Theorem \ref{thm:J} by considering the Ornstein-Uhlenbeck evolutes of the $X_i$'s (and consequently $U_m$ and $U_n$) and integrating along the semigroup. 
\end{proof}

\section{Historical Remarks}\label{sec:remarks}

Suggested by Shannon's entropy power inequality (EPI), monotonicity of entropy was a long-held conjecture that was eventually verified in 2004 when Artstein, Ball, Barthe and Naor (ABBN)  established a    `\emph{leave-one-out}' EPI for sums of independent  random variables using  a variational characterization of Fisher information \cite{artstein2004solution}. Their results imply that the Fisher information and entropy associated to  sums of independent -- but not necessarily identically distributed -- random variables enjoy a  monotonicity property that is more general than what we have proved in the present note.   Since then,    another proof of the ABBN inequality was given by Tulino and Verd\'u \cite{tulino2006monotonic} using information-estimation relationships, and Shlyakhtenko has proved a free probability extension in \cite{Shlyakhtenko}.  Finally, we note that Madiman and Barron \cite{madiman2006, madiman2007generalized}  and Madiman and Ghassemi \cite{madiman2009} have extended the ABBN  results to sums of arbitrary  subsets of independent, non-identically distributed random variables.

 It is interesting to note that the Dembo-Kagan-Shepp inequality \eqref{maxCorr} has been known since 2001, but apparently has not been connected to proving  monotonicity of entropy  until now. In retrospect, however, this connection should not be surprising.  Indeed,  all of the above referenced proofs (including that of Dembo, Kagan and Shepp \cite{dembo2001remarks}) critically hinge on variations of a  `variance drop' inequality due to Hoeffding \cite{hoeffding1948class}; once an appropriate variance drop inequality is identified,   the respective proofs and that given for Theorem \ref{thm:J} above follow a similar program.  The only notable exception in this regard is the proof of \eqref{maxCorr} by Kamath and Nair \cite{kamath2015strong}, which favors an information inequality over a variance drop inequality.  
  In any case, the brief proof of Theorem \ref{thm:J}    illustrates that monotonicity of entropy and Fisher information may be viewed as a direct consequence of the contraction  
$\EE[ \left| \EE\left[ \vartheta(S_{m})|S_n  \right]\right|^2 ] \leq \frac{m}{n}\EE[|\vartheta(S_{m})|^2]$, and may be of   interest to those familiar with the Dembo-Kagan-Shepp  maximal correlation identity, or the Kamath-Nair strong data processing result. 

Finally, we observe that the proof of \eqref{maxCorr} in \cite{kamath2015strong} goes through verbatim for random vectors, so the argument above extends immediately to the multidimensional setting.

\section*{Acknowledgements}
The author thanks Mokshay Madiman and an anonymous referee for helpful comments that improved the historical remarks.

\vskip3ex

\end{document}